\documentclass{article}

\usepackage[final]{graphicx}
\usepackage{amsfonts}
\usepackage{amssymb}
\usepackage{pbox}
\usepackage{amsmath}
\DeclareMathOperator{\tr}{tr}
\DeclareMathOperator{\diag}{diag}

\usepackage{caption}
\usepackage{subcaption}
\usepackage{mathtools}
\usepackage{braket}
\renewcommand\bra[1]{{\langle{#1}|}}
\renewcommand\ket[1]{{|{#1}\rangle}}
\usepackage{chngcntr}
\usepackage{multirow}
\usepackage{longtable}
\usepackage{tabularx}
\usepackage[labelsep=period]{caption}
\date{}
\usepackage{graphicx}
\usepackage{enumerate}
\usepackage{calc}
\usepackage{indentfirst}
\usepackage{booktabs}
\usepackage{dsfont}
\usepackage{arydshln}
\usepackage{bm}
\usepackage{amsthm}
\usepackage{makecell}
\usepackage{slashbox}
\usepackage[title]{appendix}
\usepackage{float}
\newcommand{\floor}[1]{\left\lfloor #1 \right\rfloor}

\usepackage{sectsty}
\sectionfont{\fontsize{12}{15}\selectfont}

\newtheorem{theorem}{Theorem}

\newtheorem{lemma}[theorem]{Lemma}
\newtheorem{corollary}[theorem]{Corollary}

\theoremstyle{definition}
\newtheorem{example}[theorem]{Example}

\usepackage{hyperref}

\begin{document}
\pagenumbering{arabic}

\title{\textbf{Entangling power of multipartite unitary gates}}
\author{Tomasz Linowski$^{1,2}$\footnote{linowski@cft.edu.pl}, Grzegorz Rajchel-Mieldzio\'c$^1$, Karol \.{Z}yczkowski$^{1,3}$}
\date{%
    \small $^1$Center for Theoretical Physics, Polish Academy of Sciences, Al. Lotnik{\'o}w 32/46, 02-668 Warszawa, Poland\\%
    \small $^2$International Centre for Theory of Quantum Technologies, University of Gdansk, ul. Wita Stwosza 63, 80-308 Gda{\'n}sk, Poland \\ 
    \small $^3$Smoluchowski Institute of Physics, Jagiellonian University, ul. {\L{}}ojasiewicza 11, 30-348 Krak{\'o}w, Poland\\[2ex]%
    \today
}

\maketitle

\begin{abstract}
We study the entangling properties of multipartite unitary gates with respect to the measure of entanglement called \emph{one-tangle}. Putting special emphasis on the case of three parties, we derive an analytical expression for the entangling power of an $n$-partite gate as an explicit function of the gate, linking the entangling power of gates acting on $n$-partite Hilbert space of dimension $d_1 \ldots d_n$ to the entanglement of pure states in the Hilbert space of dimension $(d_1 \ldots d_n)^2$. Furthermore, we evaluate its mean value averaged over the unitary and orthogonal groups, analyze the maximal entangling power and relate it to the absolutely maximally entangled (AME) states of a system with $2n$ parties. Finally, we provide a detailed analysis of the entangling properties of three-qubit unitary and orthogonal gates.
\end{abstract}

\section{Introduction}
Quantum entanglement is among the most important resources of modern quantum technologies. Along with quantum superposition, it is the main tool needed for superdense coding \cite{superdense_coding}, quantum teleportation \cite{quantum_teleportation}, efficient quantum tomography, measurement precision beyond classical limit \cite{measurement_heisenberg_limit}, and many other practical applications \cite{quantum_information_book,quantum_algorithms}.

In quantum computation, the computation algorithms are realized with quantum gates -- unitary transformations $U$ acting on the state-space of the system, that describe non-trivial interactions between different subsystems. In this context, a crucial property of a given gate acting on a composed quantum system is its \emph{entangling power} $\epsilon_\tau(U)$, which is defined  \cite{Zanardi_introduction} as the average entanglement created by the gate when acting on a generic fully separable pure state $\ket{\psi_{\textrm{sep}}}$:
\begin{equation} \label{eq:epsilon_states}
\begin{split}
\epsilon_\tau(U) \coloneqq \braket{\tau
	\big(U\ket{\psi_{\textrm{sep}}}\big)}
	_{\ket{\psi_{\textrm{sep}}}\in\mathcal{H}}.
\end{split}
\end{equation}
Here $\tau$ denotes the chosen measure of entanglement and the brakets $\braket{\cdot}$ represent the average over the set of pure states in the Hilbert space $\mathcal{H}$ with respect to the unique, unitarily invariant measure. For instance, characterizing the entanglement of an initially separable state by its negativity \cite{Zy99} allows one to introduce a quantity accessible in an experiment \cite{entangling_efficiency}.

In general, independently of the measure $\tau$ selected, there exists no closed-form expression for the entangling power as a function of the gate. In other words, given a certain gate~$U$, its entangling power has to be computed numerically by performing the~average~(\ref{eq:epsilon_states}), which is time-consuming and yields only approximate results.

The key exception is the class of bipartite gates of dimension $d\times d$. For these systems, Zanardi \cite{Zanardi_formula} derived an expression relating the entangling power of the gate (\ref{eq:epsilon_states}) with its linear entropy:
\begin{equation} \label{eq:epsilon_Zanardi}
\begin{split}
\epsilon_{S_L}(U) = \frac{d}{d+1}\left[S_L(U)+S_L(US)-S_L(S)\right],
\end{split}
\end{equation}
where $S_L$ is the linear entanglement entropy of the gate, determined by the purity of its Schmidt vector in the operator Schmidt decomposition, and $S$ is the SWAP gate. Note that above, the linear entropy is used in two different meanings with two different arguments: on the left as the chosen measure of entanglement of the pure state $S_L\big(\ket{\psi}\big)$, and on the right as a function of the unitary matrix $S_L(U)$. Needless to say, for most applications the expression of Zanardi is superior to the original definition~(\ref{eq:epsilon_states}), allowing one to make precise, analytical statements regarding the entangling power of bipartite quantum gates.

In this contribution, we generalize the notion of entangling power to multipartite unitary gates. Investigation of the problem was started by Scott \cite{Scott}, who derived formulae for entangling power in the case of several subsystems of an equal dimension $d$. These results are particularly important in the context of quantum correction codes \cite{codes1,codes2,codes3}. In this work, we study entangling power from a different perspective, not as a resource for a given quantum protocol, but as a physical property of a given gate, possibly acting on subsystems of different dimensions.

More recently, investigation of the problem of the entangling power of quantum gates acting on multipartite systems was pursued by Chen et al \cite{multipartite_ent_p_minimum}. They analyzed the minimal entanglement created by a given unitary gate acting on a pure product state, described by Shannon entropy of entanglement. In this work we use the linear entropy of entanglement in the form of the so-called \emph{one-tangle} measure \cite{tangles}, which allows us to derive explicit analytical expressions for entangling power of an arbitrary unitary gate.

Multiple reasons led us to consider the linear entropy of entanglement instead of other metrics, such as the Shannon / von Neumann entropy.
Firstly, the linear entropy $S_L$ is easier to compute than the von Neumann entropy $S_V$ and it is also easier to observe: it concerns the sum of the terms $p_i^2$ which have a direct experimental interpretation of a probability that a coincidence of two independent events occurs. Secondly, the linear entropy provides a direct lower bound for the Shannon entropy, $S_V \geqslant -\log(1-S_L)$ \cite{Zyczkowski2003}. Thirdly, unlike entanglement measures such as three-tangle, the linear entropy of entanglement in the form of one-tangle is applicable to all possible systems, which is crucial for our work. Finally, aside from conceptual arguments, significant rationale for adoption of the one-tangle measure was the relative simplicity of calculations.

We start by following the strategy of Zanardi \cite{Zanardi_formula} for tripartite systems of dimension $d_1d_2d_3$. We rewrite the formula (\ref{eq:epsilon_states}) for the entangling power of~a~tripartite unitary gate with respect to one-tangle in terms of explicit functions of the unitary matrix. Besides its potential to be utilized in practical applications, such as looking for optimal entangling gates, the formula offers additional physical insight into the gates' entangling power, linking it to the entanglement of states in the extended Hilbert space of dimension $(d_1d_2d_3)^2$. These results are then generalized to unitary gates acting on a system with an arbitrary number $n$ of subsystems. A thorough analysis of the entangling properties of $n$-partite gates, in particular three-qubit gates, is also provided.

This work is organized as follows. In Section \ref{sec:theory}, we briefly characterize one-tangle -- the aforementioned measure of entanglement, as well as relevant related topics. In Section \ref{sec:formula}, we present and prove the analytical formula for the entangling power as an explicit function of the matrix~$U\in U(d_1d_2d_3)$ -- see Theorem~\ref{th:epsilon_formula_geometric}. In Section~\ref{sec:tripartite_gates}, we explore the entangling properties of general tripartite gates, including the mean entangling power averaged over ensemble of random orthogonal/unitary matrices of a fixed size with respect to the Haar measure on the corresponding groups. Furthermore, we investigate the maximal entangling power and its relation to absolutely maximally entangled (AME) states \cite{AME_1,AME_2} of six-party systems. In Section \ref{sec:multipartite}, we extend all the previous results to~the general case of unitary gates acting on an arbitrary number of parties. In~Section \ref{sec:gate_properties}, we characterize the entangling properties of several relevant classes of three-qubit unitary gates. Finally, in Section \ref{sec:summary}, we summarize our findings and propose some related open problems.

\section{Tripartite entanglement}\label{sec:theory}
We begin with tripartite systems, $\mathcal{H}=\mathcal{H}_1\otimes\mathcal{H}_2\otimes\mathcal{H}_3$, $\dim\mathcal{H}_i=d_i$. As the measure of entanglement we choose the \emph{one-tangle} \cite{tangles,geometry_of_quantum_states}, defined as:
\begin{equation} \label{eq:one-tangle}
\begin{split}
\tau_1\big(\ket{\psi}\big)\coloneqq\frac{1}{3}\left[
	\tau_{12|3}\big(\ket{\psi}\big)
	+ \tau_{13|2}\big(\ket{\psi}\big)
	+ \tau_{23|1}\big(\ket{\psi}\big)\right],
\end{split}
\end{equation}
where for pure states
\begin{equation} \label{eq:tau_C}
\begin{split}
\tau_{g|g'}\big(\ket{\psi}\big)\coloneqq 
	2\left(1-\tr\big(\tr_{g}\ket{\psi}\bra{\psi}\big)^2\right)
\end{split}
\end{equation}
denotes the so-called \emph{generalized concurrence} -- a measure of entanglement with respect to the given splitting $g|g'$ of the Hilbert space. One-tangle can be thus interpreted as a measure of the total amount of entanglement in the state with respect to all bipartitions of the system. 

The range of one-tangle is defined by the range of generalized concurrence~\cite{concurrence_range}:
\begin{equation} \label{eq:tau_C_max}
\begin{split}
0\leqslant \tau_{g|g'}\big(\ket{\psi}\big) 
	\leqslant 2\frac{\min(d_g,d_{g'})-1}{\min(d_g,d_{g'})},
\end{split}
\end{equation}
where $d_g$ denotes the dimension of the partition $g$, with the former attained by separable states and the latter by maximally entangled states with respect to the bipartition $g|g'$. 

Contrary to the bipartite case, in the case of tripartite systems, there is a number of locally inequivalent classes of entangled states, and even maximally entangled states \cite{three_qubits_entanglement_types} (which shows why tripartite entanglement is significantly more complex than in the bipartite case). In particular, in the case of three qubits, there are \emph{two} such classes: the GHZ class, represented by the state
\begin{equation} \label{eq:GHZ}
\begin{split}
\ket{GHZ}\coloneqq \frac{1}{\sqrt{2}}\big(\ket{000}+\ket{111}\big),
\end{split}
\end{equation}
as well as the W class, represented by the state
\begin{equation} \label{eq:W}
\begin{split}
\ket{W}\coloneqq \frac{1}{\sqrt{3}}\big(\ket{001}+\ket{010}+\ket{100}\big).
\end{split}
\end{equation}
Intuitively, the entanglement in GHZ states can be understood as ``genuine'' tripartite entanglement. If one qubit is traced out, the resulting state is separable. The entanglement in W states, on the other hand, is more akin to bipartite entanglement. If one qubit is traced out, the resulting state is still entangled. The amount of entanglement in~the~GHZ and~the~W~states with respect to one-tangle is equal to $1$ and $8/9$, respectively, the former of which is in this case maximal and equal~to~the right hand side of inequality (\ref{eq:tau_C_max}).

We stress that one-tangle is not capable of distinguishing between different
types of entanglement. Rather, it is concerned with the total amount of entanglement in the state, which is in the main focus of this work. 
Working with one-tangle makes it possible to derive results for the entangling power of any gate acting on a tripartite as well as general multipartite systems. In scenarios in which one is interested in a specific entanglement type,
such as the GHZ type, one needs to apply other measures of entanglement, such as three-tangle \cite{tangles,geometry_of_quantum_states}.

\section{Entangling power of a given tripartite gate}\label{sec:formula}
We state the main result of this work -- an analytical formula for the entangling power of unitary matrices, in two steps. Firstly, in Lemma \ref{th:epsilon_formula_basis}, we write and prove the formula in a basis-explicit form. Then, in Theorem \ref{th:epsilon_formula_geometric}, we rephrase it in geometric terms. 
There are two reasons for such a choice: firstly, dividing the statement into two steps should make it more accessible for the reader. Secondly, and more importantly, for some purposes, including several results stated further in this work, the basis-explicit form is more practical than its geometric counterpart.

To this end, we observe that~the matrix elements of any unitary operator $U\in SU(d_1d_2d_3)$ acting in the Hilbert space $\mathcal{H}=\mathcal{H}_1\otimes\mathcal{H}_2\otimes\mathcal{H}_3$, $\dim\mathcal{H}_i=d_i$, can be conveniently written in a six-index notation as
\begin{equation} \label{eq:U_basis}
\begin{split}
U^{j_1j_2j_3}_{j_{1'}j_{2'}j_{3'}}
	\coloneqq\bra{j_1j_2j_3}U\ket{j_{1'}j_{2'}j_{3'}},
\end{split}
\end{equation}
where $j_i,{j_{i'}}\in\{0,d_i-1\}$.

Before we proceed, we note that the Einstein summation convention is used throughout the whole work, i.e. repeating indices are summed upon.

We can now state the following.

\begin{lemma}\label{th:epsilon_formula_basis}
The definition (\ref{eq:epsilon_states}) of the entangling power for a tripartite system with one-tangle (\ref{eq:one-tangle}) as the entanglement measure is equivalent to
\begin{equation} \label{eq:epsilon_definition}
\begin{split}
\epsilon_1(U)\equiv\epsilon_{\tau_1}(U)=\frac{1}{3}\left[
	\epsilon_{12|3}(U)+\epsilon_{13|2}(U)+\epsilon_{23|1}(U)\right],
\end{split}
\end{equation}
where
\begin{equation} \label{eq:epsilon_abc_basis}
\begin{split}
\epsilon_{ab|c}(U)
 \coloneqq \braket{\tau_{ab|c}
	\big(U\ket{\psi_{\textnormal{sep}}}\big)}
	_{\ket{\psi_{\textnormal{sep}}}\in\mathcal{H}} =
	2\Bigg[1-\Bigg(\prod_{i=1}^3\frac{1}{d_i(d_i+1)}\Bigg)
	u_{\vec{r}}\:u_{\vec{s}}\:u_{\vec{t}}\:
	f_{\vec{r},\vec{s},\vec{t}}^{ab|c}(U)\Bigg]
\end{split}
\end{equation}
defines the entangling power of the gate $U$ with respect to the bipartition $ab|c$. Above, $u_{\vec{v}}\coloneqq\delta^{v_1}_{v_2}\delta^{v_3}_{v_4}+\delta^{v_1}_{v_4}\delta^{v_3}_{v_2}$, while
\begin{equation} \label{eq:f_indices}
\begin{split}
f_{\vec{r},\vec{s},\vec{t}}^{ab|c}(U)&\coloneqq
	\delta^{i_a}_{l_a}\delta^{i_b}_{l_b}\delta^{i_c}_{j_c}
	\delta^{k_a}_{j_a}\delta^{k_b}_{j_b}\delta^{k_c}_{l_c}\:
	U_{r_1s_1t_1}^{i_1i_2i_3}\,
	\left(U^\dag\right)_{j_1j_2j_3}^{r_2s_2t_2}\,
	U_{r_3s_3t_3}^{k_1k_2k_3}\,
	\left(U^\dag\right)_{l_1l_2l_3}^{r_4s_4t_4}.
\end{split}
\end{equation}
We emphasize that implied summation over \emph{all} possible four-component vectors $\vec{r}$, $\vec{s}$, $\vec{t}$ is taken in eq.(\ref{eq:epsilon_abc_basis}), with vector elements spanned by $\{0,\ldots,d_i-1\}$.
\end{lemma}

\begin{proof}
Instead of averaging over states, we can average over local unitaries acting on some chosen separable state $\ket{\psi_0}$. In other words, we can rewrite definition~(\ref{eq:epsilon_states})~as
\begin{equation} \label{eq:epsilon_gates}
\begin{split}
\epsilon_\tau(U) = \braket{\tau
	\big[U \left(U_1 \otimes U_2 \otimes U_3\right)
	\ket{\psi_0}\big]}_{U_i\in SU(d_i)},
\end{split}
\end{equation}
where we have assumed a tripartite system. It is clear from the definition of one-tangle (\ref{eq:one-tangle}), that the entangling power (\ref{eq:epsilon_gates}) with one-tangle as the input is the sum of three terms of the form
\begin{equation} \label{eq:epsilon_ABC_appendix}
\begin{split}
\epsilon_{ab|c}(U)\coloneqq
	\braket{\tau_{ab|c}
	\big[U \left(U_1 \otimes U_2 \otimes U_3\right)
	\ket{\psi_0}\big]}_{U_i\in SU(d_i)}.
\end{split}
\end{equation}
All we need to do is to show that the above quantity has the conjectured form~(\ref{eq:epsilon_abc_basis}).

The proof consists of two steps. 
In the first step we compute $\tau_{ab|c}$. 
With no loss of generality, we choose the basis of the Hilbert space to be such that $\ket{\psi_0}\eqqcolon\ket{000}$ is its first element. Then,
\begin{equation}
\begin{split}
	U \left(U_1 \otimes U_2 \otimes U_3\right)
	\ket{\psi_0}=U^{a_1a_2a_3}_{b_1b_2b_3}\:
	(U_1)^{b_1}_0(U_2)^{b_2}_0(U_3)^{b_3}_0\ket{a_1a_2a_3}.
\end{split}
\end{equation}
Using this notation, one can patiently calculate $\tau_{ab|c}$ according to the definition~(\ref{eq:tau_C}): first performing the outer product $\ket{U}\bra{U}$, then the partial trace $\tr_{ab}\ket{U}\bra{U}$, next its square $\big(\tr_{ab}\ket{U}\bra{U}\big)^2$. The final result can be written as
\begin{equation} \label{eq:tau_ABC_appendix}
\begin{split}
\tau_{ab|c}\big[U \left(U_1 \otimes U_2 \otimes U_3\right)
	\ket{\psi_0}\big] = 2\bigg[1-
	(U_1)^{\vec{r}}(U_2)^{\vec{s}}(U_3)^{\vec{t}}\:
	f_{\vec{r},\vec{s},\vec{t}}^{ab|c}(U)\bigg],
\end{split}
\end{equation}
where
\begin{equation}
\begin{split}
(U_i)^{\vec{v}}&
	\coloneqq(U_i)^{v_1}_0(U_i^\dag)_{v_2}^0(U_i)^{v_3}_0(U_i^\dag)_{v_4}^0,
\end{split}
\end{equation}
while the functions $f_{\vec{r},\vec{s},\vec{t}}^{ab|c}(U)$ are defined in eq. (\ref{eq:f_indices}).

The second step is to compute $\epsilon_{ab|c}$. By definition (\ref{eq:epsilon_ABC_appendix}),
\begin{equation}
\begin{split}
\epsilon_{ab|c}(U)
	= \int dU_1 \int dU_2 \int dU_3 \; 
	\tau_{ab|c}\big[U \left(U_1 \otimes U_2 \otimes U_3\right)
	\ket{\psi_0}\big],
\end{split}
\end{equation}
where the integration is to be performed in accordance with the normalized Haar measure on the unitary group. Formula (\ref{eq:tau_ABC_appendix}) implies that
\begin{equation}\label{eq:tau_ABC_appendix_mean}
\begin{split}
\epsilon_{ab|c}(U) 
	= 2\bigg[1-
	f_{\vec{r},\vec{s},\vec{t}}^{ab|c}(U)
	\int dU_1 (U_1)^{\vec{r}}
	\int dU_2 (U_2)^{\vec{s}}
	\int dU_3 (U_3)^{\vec{t}}\bigg],
\end{split}
\end{equation}
where we have moved the function $f_{\vec{r},\vec{s},\vec{t}}^{ab|c}(U)$ in front of the integrals to emphasize that it does not depend on matrices $U_i$, $i=1,2,3$. Now, all three mutually independent integrals can be easily calculated using the handy result from \cite{unitary_integration_original,unitary_integration} regarding the integration of the second moments of random unitary matrices over the relevant unitary group:
\begin{equation} \label{eq:symbolic_integration}
\begin{split}
\int_{U(d)} dU U^{i_1}_{j_1} (U^\dag)_{i_1'}^{j_1'} U^{i_2}_{j_2} (U^\dag)_{i_2'}^{j_2'}
	= \frac{1}{d(d^2-1)}\bigg[&d\left(
	\delta^{i_1}_{i_1'}\delta^{i_2}_{i_2'}\delta^{j_1'}_{j_1}\delta^{j_2'}_{j_2}
	+ \delta^{i_1}_{i_2'}\delta^{i_2}_{i_1'}\delta^{j_2'}_{j_1}\delta^{j_1'}_{j_2}
	\right) \\
& - \left( 
	\delta^{i_1}_{i_1'}\delta^{i_2}_{i_2'}\delta^{j_2'}_{j_1}\delta^{j_1'}_{j_2}
	+ \delta^{i_1}_{i_2'}\delta^{i_2}_{i_1'}\delta^{j_1'}_{j_1}\delta^{j_2'}_{j_2}
	\right)\bigg].
\end{split}
\end{equation}
In the case at hand,
\begin{equation}
\begin{split}
\int_{U(d_i)} dU_i (U_i)^{\vec{v}}=\frac{1}{d_i(d_i+1)}\big(
\delta^{v_1}_{v_2}\delta^{v_3}_{v_4}+\delta^{v_1}_{v_4}\delta^{v_3}_{v_2}\big).
\end{split}
\end{equation}
Substituting into eq. (\ref{eq:tau_ABC_appendix_mean}), we immediately obtain the conjectured formula~(\ref{eq:epsilon_abc_basis}).
\end{proof}

In order to restate formula (\ref{eq:epsilon_abc_basis}) in geometric terms, we observe that to~every unitary operator $U\in SU(d_1d_2d_3)$ acting in the Hilbert space $\mathcal{H}=\mathcal{H}_1\otimes\mathcal{H}_2\otimes\mathcal{H}_3$, $\dim\mathcal{H}_i=d_i$, there corresponds a pure state $\ket{U}$ in the extended Hilbert space $\mathcal{H}\otimes\mathcal{H}'$, $\mathcal{H}'=\mathcal{H}_{1'}\otimes\mathcal{H}_{2'}\otimes\mathcal{H}_{3'}$, $\dim\mathcal{H}_{i'}=\dim\mathcal{H}_{i}=d_i$. This is essentially an application of the Choi-Jamio{\l}kowski isomorphism \cite{jamiolkowski,choi} to the special case of unitary operations.

Given a basis $\{\ket{j_1j_2j_3}\}$ of the Hilbert space $\mathcal{H}$, the coefficients of the state $\ket{U}\in\mathcal{H}\otimes\mathcal{H}'$ are defined by the following relation:
\begin{equation} \label{eq:U_state}
\begin{split}
\ket{U}\coloneqq\frac{1}{\sqrt{d_1d_2d_3}}U^{j_1j_2j_3}_{j_{1'}j_{2'}j_{3'}}
	\ket{j_1j_2j_3j_{1'}j_{2'}j_{3'}},
\end{split}
\end{equation}
where the dimensional factor ensures proper normalization. Note that the equality (\ref{eq:U_state}) is valid in any basis, so the association $U\to\ket{U}$ may be viewed as a geometric statement.

We can now restate Lemma \ref{th:epsilon_formula_basis} in geometric terms:

\begin{theorem} \label{th:epsilon_formula_geometric}
Definition (\ref{eq:epsilon_states}) of the entangling power for a tripartite system with one-tangle (\ref{eq:one-tangle}) as the entanglement measure is given by eq. (\ref{eq:epsilon_definition}) with the~entangling power of the gate $U$ on the bipartition $ab|c$ (\ref{eq:epsilon_abc_basis}) equal to
\begin{equation} \label{eq:epsilon_abc_geometric}
\begin{split}
\epsilon_{ab|c}(U)=
	2\left[1-\Bigg(\prod_{i=1}^3\frac{d_i}{d_i+1}\Bigg)
	\sum_{x'|y'}
	\tr\big(\tr_{abx'}\ket{U}\bra{U}\big)^2\right].
\end{split}
\end{equation}
In the above expression, the summation is over all \emph{ordered} bipartitions $x'|y'$~of~$\mathcal{H}'$:
\begin{equation} \label{eq:xy_range}
\begin{split}
x'|y'\in\big\{1'2'3'|\cdot,\;1'2'|3',\;1'3'|2',\;2'3'|1',\;1'|2'3',\;2'|1'3',\;3'|1'2',\;\cdot|1'2'3'\big\},
\end{split}
\end{equation}
including the two trivial bipartitions $1'2'3'|\cdot$ and $\cdot|1'2'3'$, where the dot denotes an empty set, while $\ket{U}$ is the state (\ref{eq:U_state}) associated with the matrix $U$.
\end{theorem}

\begin{proof}
The proof consists of a relatively lengthy but straightforward direct calculation of the expression (\ref{eq:epsilon_abc_geometric}) in any chosen basis and comparison with the~formula (\ref{eq:epsilon_abc_basis}).
\end{proof}

It is worth emphasizing that while the auxiliary Hilbert space $\mathcal{H}'$ is formally an identical copy of the original Hilbert space $\mathcal{H}$, from the point of view of this work, the two play different roles. In particular, it makes no sense to consider any divisions of $\mathcal{H}$ other than the only three unique, ordered bipartitions $ab|c\in\{12|3,\;13|2,\;23|1\}$. However, in order to properly account for all the divisions of the total Hilbert space $\mathcal{H}\otimes\mathcal{H}'$, it is necessary to consider all eight, unordered bipartitions $x'|y'$ of $\mathcal{H}'$ as in eq. (\ref{eq:xy_range}). We stress that this convention is used throughout the rest of this work.

\section{Statistical properties of ensembles of tripartite gates}\label{sec:tripartite_gates}
After establishing an explicit formula (\ref{eq:epsilon_abc_geometric}) for the entangling power of tripartite gates, we would like to explore some of its consequences regarding the entangling properties of general tripartite unitary gates.

We begin with a formula for the upper bound for the entangling power.

\begin{theorem} \label{th:epsilon_max}
The entangling power of tripartite unitary gates $U\in U(d_1d_2d_3)$ is bounded from above by
\begin{equation} \label{eq:epsilon_max_conjecture}
\begin{split}
\tilde{\epsilon}_{1}\coloneqq\frac{1}{3}\left(
    \tilde{\epsilon}_{12|3}+\tilde{\epsilon}_{13|2}+\tilde{\epsilon}_{23|1}\right),
\end{split}
\end{equation}
where $\tilde{\epsilon}_{ab|c}$ is the upper bound for the~entangling power of the gate $U$ on the bipartition $ab|c$ (\ref{eq:epsilon_abc_geometric}):
\begin{equation} \label{eq:epsilon_max}
\begin{split}
\tilde{\epsilon}_{ab|c}\coloneqq 2
	-2\Bigg(\prod_{i=1}^3\frac{d_i}{d_i+1}\Bigg)
	\bigg[8-\sum_{x'|y'}
	\frac{\min(d_{abx'},d_{cy'})-1}{\min(d_{abx'},d_{cy'})}\bigg],
\end{split}
\end{equation}
where the summation is over $x'|y'$ as in (\ref{eq:xy_range}) and $d_{abx'}$, $d_{cy'}$ denote the dimensions of the respective bipartitions of the extended Hilbert space $\mathcal{H}\otimes\mathcal{H}'$.
\end{theorem}

\begin{proof}
It is immediate to see from the definition of the entangling power (\ref{eq:epsilon_definition}) that if the quantity (\ref{eq:epsilon_max}) is indeed the upper bound for the the~entangling power (\ref{eq:tau_C}) of the gate $U$ on the bipartition $ab|c$ (\ref{eq:epsilon_abc_geometric}), the theorem is true. All that remains is to prove this assumption.

Let us add and substract the following quantity
\begin{equation}
\begin{split}
2\Bigg(\prod_{i=1}^3\frac{d_i}{d_i+1}\Bigg)\sum_{x'|y'}\left[1+
	\frac{\min(d_{abx'},d_{cy'})-1}{\min(d_{abx'},d_{cy'})}\right]
\end{split}
\end{equation}
from the right hand side of eq. (\ref{eq:epsilon_abc_geometric}). After regrouping the terms, we arrive at
\begin{equation} \label{eq:epsilon_interpretation}
\begin{split}
\epsilon_{ab|c}(U)=\tilde{\epsilon}_{ab|c}
	-\Bigg(\prod_{i=1}^3\frac{d_i}{d_i+1}\Bigg)\sum_{x'|y'}
	\left[2\frac{\min(d_{abx'},d_{cy'})-1}{\min(d_{abx'},d_{cy'})}
	-\tau_{abx'|cy'}\big(\ket{U}\big)\right].
\end{split}
\end{equation}
Recall that in accordance with eq. (\ref{eq:tau_C_max}), the quantity $2\frac{\min(d_{abx'},d_{cy'})-1}{\min(d_{abx'},d_{cy'})}$ is the upper bound for $\tau_{abx'|cy'}$, reachable only by states that are maximally entangled with respect to the bipartition $abx'|cy'$. This means that each and every term in the above sum, and thus the sum itself, is non-negative. Since the sum enters the equation with a minus sign, $\epsilon_{ab|c}(U)$ is at most equal to the quantity $\tilde{\epsilon}_{ab|c}$. This completes the proof.
\end{proof}

Besides bounding the maximum value of the entangling power of unitary gates $U\in U(d_1d_2d_3)$ from above, the theorem provides us with an interpretation of the formula for the entangling power. Looking at eq. (\ref{eq:epsilon_interpretation}) we can see that the value $\epsilon_{ab|c}(U)$ is the highest when the values $\tau_{abx'|cy'}(\ket{U})$ are the highest. The~entangling power of the tripartite gate $U$ is proportional to the total amount of entanglement in the sixpartite state $\ket{U}$. 

Note that Theorem \ref{th:epsilon_max} provides only an upper bound for the maximum value of the entangling power of unitary gates $U\in U(d_1d_2d_3)$. This bound \emph{may not be tight} in general. It follows immediately from our discussion that the bound is tight if \emph{and only if} the maximizing $U$ fulfills
\begin{equation} 
\begin{split}
\tau_{abx'|cy'}\big(\ket{U}\big)=2\frac{\min(d_{abx'},d_{cy'})-1}{\min(d_{abx'},d_{cy'})}
\end{split}
\end{equation}
for all bipartitions $abx'|cy'$ of $\mathcal{H}\otimes\mathcal{H}'$ entering the formula for the entangling power. In fact, further inspection reveals that for the bound to be tight the above relation must be true also for all the other bipartitions. This is either because of the symmetric property of generalized concurrence: $\tau_{abx'|cy'}=\tau_{cy'|abx'}$, or the origin of the state $\ket{U}$ as a unitary matrix. In other words, $U$ must be such that the state $\ket{U}$ is a maximally entangled state with respect to each of the bipartitions.

Such states are known in the literature as \emph{absolutely maximally entangled states} (AME) \cite{AME_1,AME_2,AME_state}. They can exist only in multipartite quantum systems where each subsystem has the same dimension \cite{AME_2}. The set of all such states in an $n$-partite Hilbert space of subsystem dimension $d$ is denoted by AME($n$,$d$). Curiously, there exist pairs $(n,d)$, for which AME states do not exist. For instance, there are no such states for four- and seven-qubit systems \cite{AME_four_qubits,AME_seven_qubits}. However, it has been shown \cite{ame6dexist_1,ame6dexist_2}, that the set AME($6$,$d$), which corresponds to tripartite gates studied here, is non-empty for all dimensions $d$.

Because of the property of the AME states called \emph{multiunitarity} \cite{AME_state}, utilizing the~recipe~(\ref{eq:U_state}), one can use the AME states to construct a unitary matrix $U$ maximizing the entangling power $\epsilon_{1}$ for tripartite gates acting on $\mathcal{H}^{\otimes 3}_d$. A link between large entangling power of unitary operators and strongly entangled multipartite states in an extended space
was already discussed by Scott in \cite{Scott}. We are now in position to establish a more precise relation and
propose here the following result.

\begin{corollary} \label{th:AME}
The upper bound $\tilde{\epsilon}_{1}$ for the maximum entangling power of unitary gates $U(d_1d_2d_3)$ given in Theorem \ref{th:epsilon_max} is tight if and only if $d_1=d_2=d_3\equiv d$. 

Furthermore, given a state $\ket{\psi(6,d)}\in\textrm{AME}(6,d)$ the matrix elements of the maximizing $U$ can be recovered using the recipe (\ref{eq:U_state}), explicitly
\begin{equation}
\begin{split}
U^{j_1j_2j_3}_{j_{1'}j_{2'}j_{3'}}
	=\sqrt{d^3}\braket{j_1j_2j_3j_{1'}j_{2'}j_{3'}|\psi(6,d)}.
\end{split}
\end{equation}
\end{corollary}
It is worth adding that one can generate states $\ket{\psi(6,d)}\in\textrm{AME}(6,d)$, and~therefore gates $U$ that maximize entangling power, using known algorithms \cite{ame6dexist_2,ame_table}.

As our last general result, we compute the mean entangling power over the unitary group $U(d_1d_2d_3)$ and the (real) orthogonal group $O(d_1d_2d_3)$.

\begin{theorem} \label{th:epsilon_mean}
Mean entangling power of tripartite unitary gates averaged over the unitary group $U(d_1d_2d_3)$ and the orthogonal group $O(d_1d_2d_3)$ with respect to the Haar measure read
\begin{equation} \label{eq:epsilon_mean}
\begin{split}
\braket{\epsilon_{1}}_{U(d_1d_2d_3)}
	= \frac{2A}{3(d_1d_2d_3+1)},
\end{split}
\end{equation}
\begin{equation} \label{eq:epsilon_mean_orthog}
\begin{split}
\braket{\epsilon_{1}}_{O(d_1d_2d_3)}
	=	\frac{2A\left(\left[\prod_{i=1}^3 d_i(d_i+1)\right]-8\right)}
	{3(d_1d_2d_3-1)(d_1d_2d_3+2)\left[\prod_{i=1}^3(d_i+1)\right]},
\end{split}
\end{equation}
where $A\coloneqq 3-d_1-d_2-d_3-d_1d_2-d_1d_3-d_2d_3+3d_1d_2d_3$.
\end{theorem}

\begin{proof}
The proof relies on the expression (\ref{eq:epsilon_abc_basis}). Since the functions $f^{ab|c}_{\vec{r},\vec{s},\vec{t}}(U)$ are given (\ref{eq:f_indices}) in terms of the second moments of the unitary matrix, one can integrate the basis-dependent expression using the previously utilized formula (\ref{eq:symbolic_integration}), which results in the conjectured expression (\ref{eq:epsilon_mean}) for the unitary group.

In the case of the orthogonal group, it is possible to find an analog of the formula (\ref{eq:symbolic_integration}) for orthogonal matrices -- we refer the reader to Appendix \ref{app:integration} for details. The proof is then fully analogous to the unitary case.
\end{proof}

We conclude the section by applying Theorems \ref{th:epsilon_max} and \ref{th:epsilon_mean} to the special case of evenly divided tripartite systems $d_1=d_2=d_3=d$, called three-qu$d$it systems, in which the expressions given therein simplify significantly.

\begin{corollary} \label{th:corollary}
The mean entangling power of three-qu$d$it gates $U\in U(d^3)$ reads
\begin{equation} \label{eq:mean_qudits}
\begin{split}
\braket{\epsilon_{1}}_{U(d^3)}=2\frac{(d-1)^2}{d^2-d+1},
\end{split}
\end{equation}
in the case of the unitary group $U(d^3)$, and
\begin{equation} \label{eq:mean_qudits_orthog}
\begin{split}
\braket{\epsilon_{1}}_{O(d^3)}
	= 2\frac{d^3(d+1)^3(d-1)-8(d-1)}{(d^3+2)(d^2+d+1)(d+1)^2},
\end{split}
\end{equation}
in the case of the orthogonal group $O(d^3)$. Furthermore, the maximum value of the entangling power is equal to
\begin{equation} \label{eq:max_qudits}
\begin{split}
\tilde{\epsilon}_{1}=2\frac{d^2+d-2}{(1+d)^2}.
\end{split}
\end{equation}
\end{corollary}

Figure \ref{fig:epsilon_comparison} is a comparison of the maximum value (\ref{eq:max_qudits}) of the entangling power of unitary three-qu$d$it gates, denoted by orange squares, with the general upper bound, equal to one-tangle of the maximally entangled state
\begin{equation} \label{eq:max_theoretical}
\begin{split}
\max\tau_1=2\frac{d-1}{d},
\end{split}
\end{equation}
denoted by green diamonds. For reference, the mean entangling power (\ref{eq:mean_qudits}) of three-qu$d$it gates has also been plotted as blue circles. As seen, while for large dimension $d$ the three quantities practically converge, for smaller dimensions the maximum is noticeably smaller than the theoretical upper bound. 

It is also worth emphasizing, that the limiting value of the three quantities is $2$. Since in our normalization this is the limiting value of one-tangle of the maximally entangled state [see eq. (\ref{eq:max_theoretical})], this means that in large dimensions, on average, the action of a typical tripartite unitary gate on a separable state results in a state close to the maximally entangled one. This may be seen as a consequence of the fact that unitary gates correspond to quantum states [eq. (\ref{eq:U_state})] and the fact that the in large dimensions generic quantum states tend to be highly entangled \cite{most_states_are_highly_entangled}.

\begin{example} \label{th:example_tripartite}
As an example, we consider three of the most popular three-qubit quantum gates \cite{three_qubit_gates}: the Fredkin gate $U_{F}$ (also known as controlled SWAP), the Toffoli gate $U_{T}$ (also known as controlled-controlled NOT) and the Deutsch gate $U_{D}(\theta)$ (a controlled-controlled complex rotation by angle $\theta$) -- see \cite{three_qubit_gates} for explicit matrix forms. We find that
\begin{equation}
\begin{split}
\epsilon_1\big(U_{D}(\theta)\big)=\frac{1}{27}[7-3\cos(2\theta)]\leqslant
    \epsilon_1(U_{F})=\epsilon_1(U_{T})=\frac{10}{27}\approx 0.37,
\end{split}
\end{equation}
which is much smaller than the mean value $\braket{\epsilon_{1}}_{U(2^3)}=2/3\approx 0.67$, as calculated using Corollary \ref{th:corollary}. We can thus conclude that all three gates are on average much less entangling than a random three-qubit unitary gate.
\end{example} 

\begin{figure}[!tb]
	\centering
	\includegraphics[width=1\textwidth]{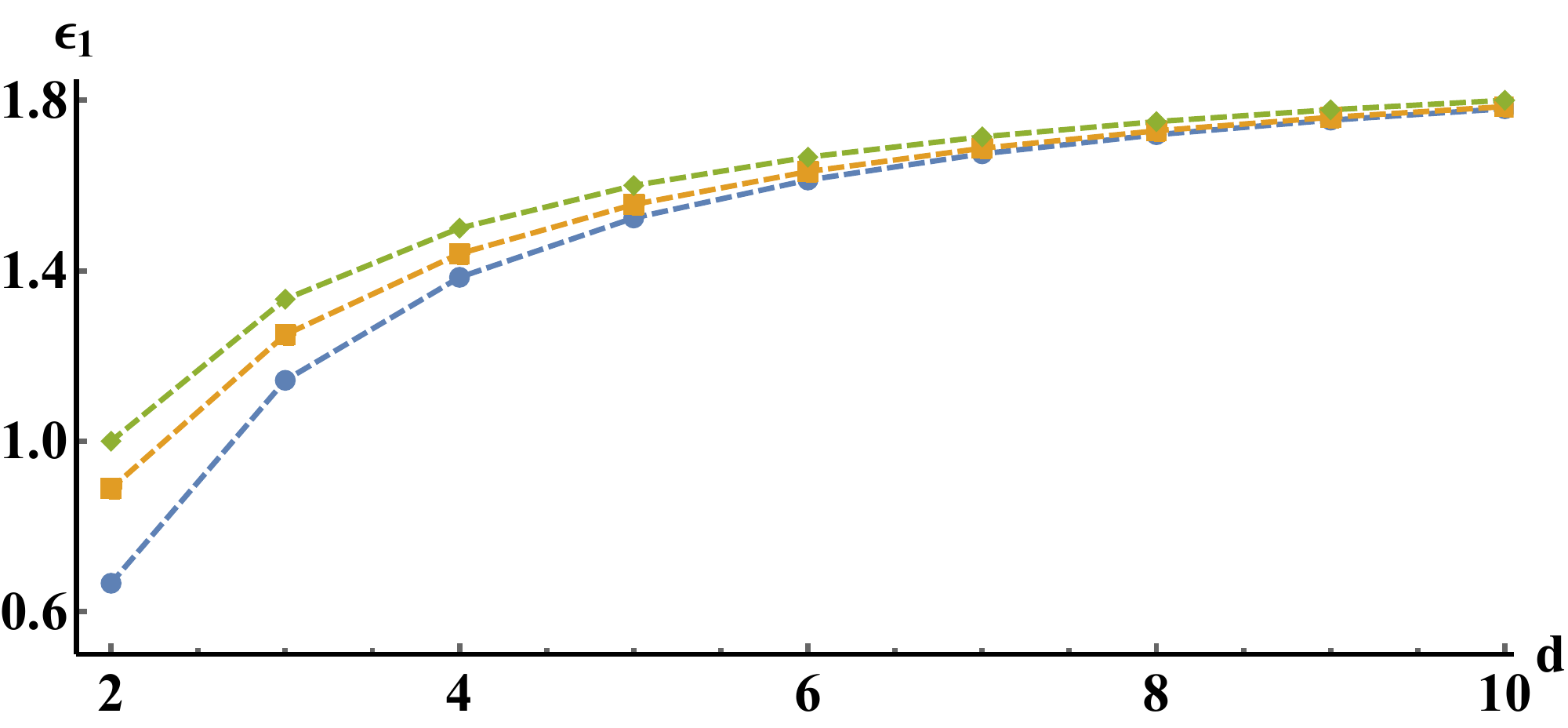}
	\captionsetup{width=1\textwidth}
	\caption{Comparison of the upper bound (\ref{eq:max_qudits}) for the entangling power of three qu$d$it unitary matrices, denoted by squares (orange), with the theoretical maximum (\ref{eq:max_theoretical}), denoted by diamonds (green), and the mean entangling power (\ref{eq:mean_qudits}), denoted by circles (blue). Dashed lines connecting the points have been plotted to guide the eye.}
	\label{fig:epsilon_comparison}
\end{figure}

\section{Multipartite unitary gates} \label{sec:multipartite}
In this section, we generalize the results from Sections \ref{sec:formula}-\ref{sec:tripartite_gates} to an arbitrary number of partitions. With the exception of Lemma \ref{th:epsilon_formula_basis}, which is skipped due to its auxiliary nature, each of the theorems provided therein is reiterated here in a version for $n$ partitions. The proofs are omitted, as they follow the same lines of reasoning as their tripartite counterparts.

In full analogy to the tripartite case, to every unitary matrix $U$ acting in the Hilbert space $\mathcal{H}=\bigotimes_{i=1}^n\mathcal{H}_i$ we relate the state $\ket{U}$ in the extended Hilbert space $\mathcal{H}\otimes\mathcal{H}'$, where $\mathcal{H}'=\bigotimes_{i=1}^n\mathcal{H}_{i'}$ and $\dim\mathcal{H}_i=\dim\mathcal{H}_{i'}=d_i$ using the following equation:
\begin{equation} \label{eq:U_state_n}
\begin{split}
\ket{U}&\coloneqq\frac{1}{\sqrt{d_1\ldots d_n}}U^{j_1\ldots j_n}_{j_{1'}\ldots j_{n'}}
	\ket{j_1\ldots j_nj_{1'}\ldots j_{n'}}.
\end{split}
\end{equation}
The $2n$-index notation has been used here to denote matrix elements,
\begin{equation}
\begin{split}
U^{j_1 \ldots j_n}_{j_{1'} \ldots j_{n'}}&\coloneqq\bra{j_1 \ldots j_n}U\ket{j_{1'} \ldots j_{n'}}.
\end{split}
\end{equation}

Furthermore, we establish the following summation convention. Summation over $p|q$ is understood as a summation over all \emph{non-trivial, unordered} bipartitions of $\mathcal{H}$. Summation over $x'|y'$, on the other hand, is understood as a summation over all \emph{ordered} bipartitions of $\mathcal{H}'$, including trivial cases. For example, in the simplest case of two partitions, $n=2$, we have
\begin{equation}
\begin{split}
p|q\in\{1|2\},\qquad x'|y'\in\{1'2'|\cdot,\;1'|2',\;2'|1',\;\cdot|1'2'\},
\end{split}
\end{equation}
where a single dot represents an empty set.
As the measure of entanglement, we choose the natural generalization of~one-tangle (\ref{eq:one-tangle}) to $n$ parties,
\begin{equation} \label{eq:one-tangle_n}
\begin{split}
\tau_{1}(U)=\frac{1}{2^{n-1}-1}\sum_{p|q}\tau_{p|q}(U),
\end{split}
\end{equation}
where $2^{n-1}-1$ is the number of bipartitions and $\tau_{p|q}$ is defined in~eq.~(\ref{eq:tau_C}). We are now ready to present the results concerning unitary gates acting on $n$-partite systems.

\begin{theorem}
Definition (\ref{eq:epsilon_states}) of the entangling power for an $n$-partite system with the the entanglement measure given by the $n$-particle generalization of one-tangle (\ref{eq:one-tangle_n}) is equivalent to
\begin{equation}
\begin{split}
\epsilon_{1}(U)=\frac{1}{2^{n-1}-1}\sum_{p|q}\epsilon_{p|q}(U),
\end{split}
\end{equation}
where
\begin{equation} \label{eq:epsilon_C_pq}
\begin{split}
\epsilon_{p|q}(U)=
	2\left[1-\Bigg(\prod_{i=1}^n\frac{d_i}{d_i+1}\Bigg)
	\sum_{x'|y'}
	\tr\big(\tr_{px'}\ket{U}\bra{U}\big)^2\right].
\end{split}
\end{equation}
denotes the entangling power of the matrix $U$ with respect to the bipartition $p|q$ of $\mathcal{H}$.
\end{theorem}

\begin{theorem} \label{th:epsilon_max_n}
The entangling power of $n$-partite unitary gates $U\in U(d_1\ldots d_n)$ is bounded from above by
\begin{equation}
\begin{split}
\tilde{\epsilon}_{1}\coloneqq\frac{1}{2^{n-1}-1}\sum_{p|q}\tilde{\epsilon}_{p|q}(U),
\end{split}
\end{equation}
where $\tilde{\epsilon}_{p|q}$ is the upper bound for the entangling power of the matrix $U$ on the bipartition $p|q$:
\begin{equation}
\begin{split}
\tilde{\epsilon}_{p|q}\coloneqq 2
	-2\Bigg(\prod_{i=1}^n\frac{d_i}{d_i+1}\Bigg)
	\left[2^n-\sum_{x'|y'}
	\frac{\min(d_{px'},d_{qy'})-1}{\min(d_{px'},d_{qy'})}\right].
\end{split}
\end{equation}
In the above expression, $d_{px'}$, $d_{qy'}$ denote the dimensions of the respective bipartitions of the extended Hilbert space $\mathcal{H}\otimes\mathcal{H}'$.
\end{theorem}

\begin{corollary}
The upper bound $\tilde{\epsilon}_{1}$ for the maximum entangling power of unitary gates $U(d_1\ldots d_n)$ given in Theorem \ref{th:epsilon_max_n} is tight if and only if $d_1=\ldots=d_n\equiv d$ and the set AME($2n$,$d$) is non-empty. 

Furthermore, given a state $\ket{\psi(2n,d)}\in\textrm{AME}(2n,d)$ the matrix elements of the unitary gate $U$ maximizing the entangling power can be recovered explicitly using the recipe (\ref{eq:U_state_n}), explicitly
\begin{equation}
\begin{split}
U^{j_1\ldots j_n}_{j_{1'}\ldots j_{n'}}
	=\sqrt{d^n}\braket{j_1\ldots j_nj_{1'}\ldots j_{n'}|\psi(2n,d)}.
\end{split}
\end{equation}
\end{corollary}

We stress that contrary to the tripartite case, in general the set AME($2n$,$d$) is not always non-trivial. A table on the existence of AME states along with construction algorithms can be found in \cite{ame_table}.

\begin{theorem} \label{th:epsilon_mean_n}
Mean entangling power of tripartite unitary gates averaged over the unitary group $U(d_1\ldots d_n)$ and the orthogonal group $O(d_1\ldots d_n)$ with respect to the Haar measure read
\begin{equation} \label{eq:mean_unitary_n}
\begin{split}
\braket{\epsilon_{1}}_{U(d_1\ldots d_n)}
	= 2\left[1-\left(\prod_{i=1}^n\frac{1}{d_i+1}\right)
	\frac{BC}
	{(2^{n-1}-1)(D+1)}\right],
\end{split}
\end{equation}
\begin{equation}
\begin{split}
\braket{\epsilon_{1}}_{O(d_1\ldots d_n)}
	=2\left[1-\left(\prod_{i=1}^n\frac{1}{d_i+1}\right)
	\frac{2^n(D+1)-2B+\frac{BD-2^n}{2^{n-1}-1}C}
	{(D-1)(D+2)}\right],
\end{split}
\end{equation}
where $B\coloneqq \sum_{i_1,\ldots,i_n=0}^1 d_1^{i_1}\ldots d_n^{i_n}$, $C\coloneqq\sum_{p|q}(d_p+d_q)$ and $D\coloneqq d_1\ldots d_n$.
\end{theorem}

Observe that if we set $n=2$, the expression (\ref{eq:mean_unitary_n}) is equivalent to the result of Zanardi for bipartite systems \cite{Zanardi_introduction}, up to the multiplicative factor $1/2$ due to different normalizations of generalized concurrence used.

\begin{corollary} \label{th:corollary_n}
The mean entangling power of $n$-qu$d$it gates $U\in U(d^n)$ reads
\begin{equation} \label{eq:mean_qudits_n}
\begin{split}
\braket{\epsilon_{1}}_{U(d^n)}=\frac{2^n(d^n+1)-2(d+1)^n}{(2^{n-1}-1)(d^n+1)}
\end{split}
\end{equation}
when averaged over the unitary group $U(d^n)$, and
\begin{equation} \label{eq:mean_qudits_orthog_n}
\begin{split}
\braket{\epsilon_{1}}_{O(d^n)}
	= \frac{\left[2^n(d^n+1)-2(d+1)^n\right]\left[d^n(d+1)^n-2^n\right]}
	{(2^{n-1}-1)(d^{2n}+d^n-2)(d+1)^n}
\end{split}
\end{equation}
when averaged over the orthogonal group $O(d^n)$. Furthermore, this entangling power is bounded from above by
\begin{equation} \label{eq:bound_qudits_n}
\begin{split}
    \tilde{\epsilon}_1 = 2 - \frac{2d^{n}}{(d+1)^{n}}\left[2^n-\frac{1}{(2^{n-1}-1)}
    \sum_{j=0}^{n}{{n}\choose{j}}\sum_{l=1}^{\floor{n/2}}{{n}\choose{l}}
    \frac{d^{n-|l-j|}-1}{d^{n-|l-j|}\:2^{\delta^l_{n/2}}}\right],
\end{split}
\end{equation}
where $\floor{\cdot}$ denotes the floor function.
\end{corollary}

Let us make three remarks regarding the above corollary.

Firstly, we note that the result (\ref{eq:mean_qudits_n}) regarding the mean over $n$-qu$d$it unitary group could also be obtained by considering an appropriately weighted average of means of entangling power classes introduced in \cite{Scott}.

Secondly, we observe that due to the effect of concentration of measure investigated recently in the context of the entangling power of  random unitary matrices \cite{multipartite_ent_p_minimum}, one can expect that for large dimension $d$ or the number of parties $n$ the value of the entangling power for a given unitary gate $U$ will typically be close to the averaged values derived above.

\begin{figure}[!tb]
	\centering
	\includegraphics[width=1\textwidth]{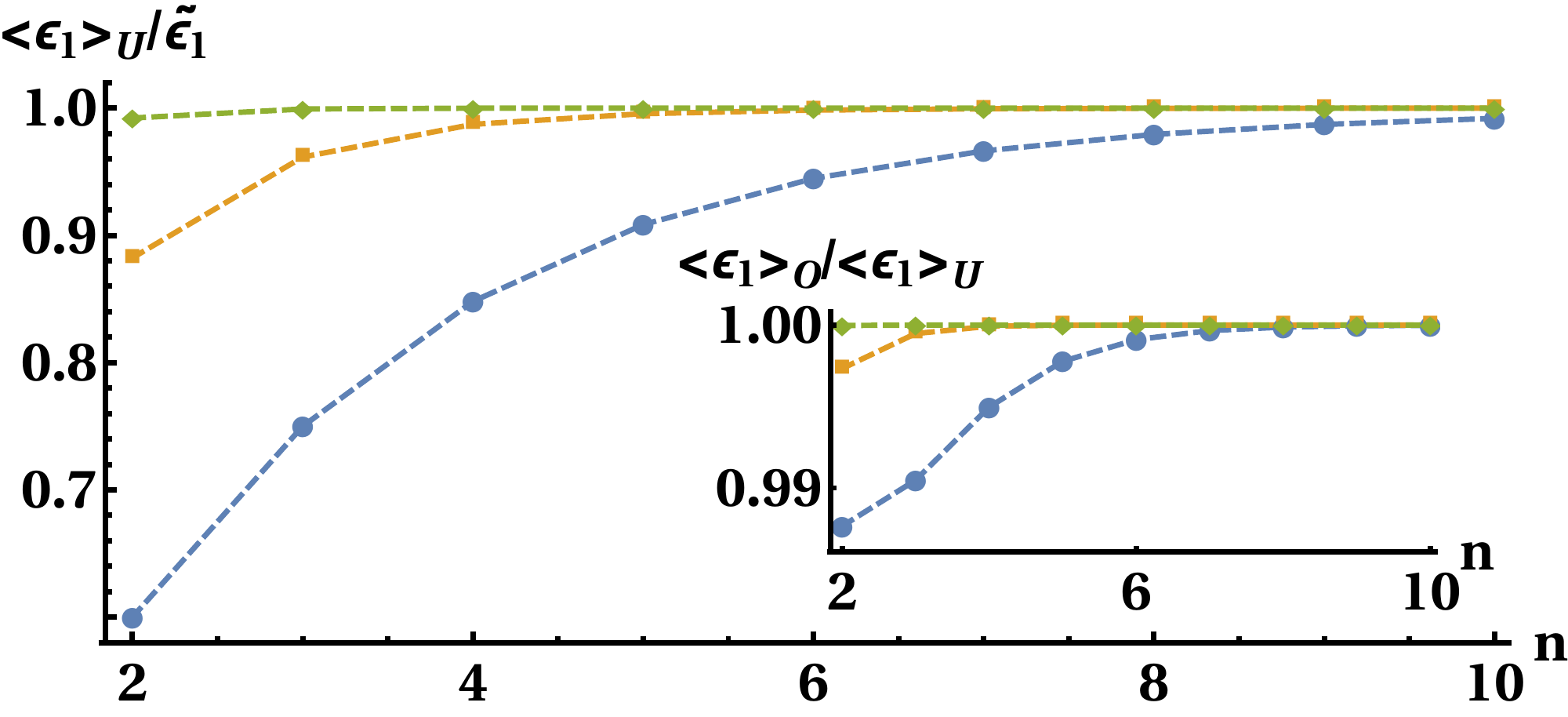}
	\captionsetup{width=1\textwidth}
	\caption{The ratio $\braket{\epsilon_{1}}_{U(d^n)}/\tilde{\epsilon}_1$ as a function of the number of partitions $n$ for partition dimensions $d=2$ -- circles (blue), $d=4$ -- squares (orange), $d=16$ -- diamonds (green). Lines joining the points have been plotted to guide the eye. In the inset: an analogous plot of the ratio $\braket{\epsilon_{1}}_{O(d^n)}/\braket{\epsilon_{1}}_{U(d^n)}$.}
	\label{fig:means_bound_comparison}
\end{figure}

Finally, we note that aside from the domain where $d$, $n$ are ``small'', the~mean values (\ref{eq:mean_qudits_n}), (\ref{eq:mean_qudits_orthog_n}) over the unitary and orthogonal groups are nearly identical. Indeed, one can check that
\begin{equation}
\begin{split}
\lim_{n\to\infty}\frac{\braket{\epsilon_{1}}_{O(d^n)}}{\braket{\epsilon_{1}}_{U(d^n)}}
    =\lim_{d\to\infty}\frac{\braket{\epsilon_{1}}_{O(d^n)}}{\braket{\epsilon_{1}}_{U(d^n)}}=1.
\end{split}
\end{equation}
Moreover, numerical simulations suggest similar results regarding the upper bound (\ref{eq:bound_qudits_n}) for the entangling power $\tilde{\epsilon}_1$. In other words, in any $n$-qu$d$it system with a sufficiently large dimension $d$ or number of partitions $n$, the values of the quantities $\braket{\epsilon_{1}}_{O(d^n)}$, $\braket{\epsilon_{1}}_{U(d^n)}$ and $\tilde{\epsilon}_1$ are practically indistinguishable. A~visual comparison of the three quantities, which further supports our claim, has been provided in Figure \ref{fig:means_bound_comparison}.

The above results generalize our findings for tripartite systems, showing that for higher dimensions a generic unitary matrix acting on $n$-partite systems is characterized by entangling power close to the maximal one. Hence, the corresponding random state determined in eq. (\ref{eq:U_state_n}) becomes asymptotically close the AME state of $2n$ parties.

\begin{example} \label{th:example_multipartite}
In \cite{G_gate}, the entangling properties of the $n$-qubit unitary gate
\begin{equation} \label{eq:diagonal_parametrization}
\begin{split}
G_n(\alpha)\coloneqq\diag\big(1,\ldots,1,e^{i\alpha}\big)
\end{split}
\end{equation}
have been analyzed. Note that $G_n(\alpha)$ is a generalization of the controlled sign gate, $G_n(\pi)=\diag\big(1,\ldots,1,-1\big)$.

Using our formalism, we were able to calculate the entangling power of the gate
\begin{equation} \label{eq:diagonal_parametrization}
\begin{split}
\epsilon_1\big(G_n(\alpha)\big)=c_n (1-\cos\alpha) \in [0,2c_n],
\end{split}
\end{equation}
where
\begin{equation} \label{eq:diagonal_parametrization}
\begin{split}
c_n=\frac{2^3}{6^n(2^n-1)}\sum_{p|q}
    \big(3^{n_p} - 2^{n_p}\big)
    \big(3^{n_q} - 2^{n_q}\big) 
\end{split}
\end{equation}
and $n_p$ and $n_q$ denote the number of qubits in partitions $p$, $q$. As seen, the closer the parameter $\alpha$ is to $\pi$, the more entangling the gate. This is, of course, exactly what one should expect, as $G_n(\pi)$ is intuitively the furthest from the non-entangling identity gate $G_n(0)$. In the particular case of three qubits, $G_3(\pi)=10/27\approx 0.37$.
\end{example} 

\section{Entangling properties of three-qubit unitary gates}\label{sec:gate_properties}
In Example \ref{th:example_tripartite}, we have considered three three-qubit gates: the Fredkin, Toffoli, and Deutsch gates. In this section, we generalize these results and characterize the entangling properties of generic three-qubit gates. We consider the following four classes of such gates: 
\begin{itemize}
\item permutation matrices $\mathcal{P}(8)$,
\item diagonal unitary matrices $\mathcal{D}(8)$,
\item unitary matrices $U(8)$,
\item orthogonal matrices $O(8)$.
\end{itemize}
We mention that the latter three are also known in the literature \cite{circular_ensembles_original,circular_ensembles_more} as the Circular Poissonian Ensemble (CPE), Circular Unitary Ensemble (CUE) and Circular Real Ensemble (CRE), respectively.

\paragraph{Permutations $\bm{\mathcal{P}(8)}$.} There are exactly $8!=40\:320$ three-qubit permutation matrices $P\in\mathcal{P}(8)$. Since this number is finite and relatively small, it is possible to calculate the entangling power (\ref{eq:epsilon_definition}) of every permutation. This yields exactly 21 different entangling classes, ranging from $0$ to $\max_{\mathcal{P}(8)}\epsilon_{1}=64/81\approx 0.79$, with mean $\braket{\epsilon_{1}}_{\mathcal{P}(8)}=\frac{184}{315}\approx 0.58$. The classification of permutations with respect to their entangling power is provided in~Table~\ref{tab:permutations} in Appendix \ref{app:permutations}.

\paragraph{Diagonal unitary matrices $\bm{\mathcal{D}(8)}$.} Any diagonal unitary $D\in \mathcal{D}(8)$ can be written as 
\begin{equation} \label{eq:diagonal_parametrization}
\begin{split}
D_{\vec{\varphi}}\coloneqq\diag\big(e^{i\varphi_1},\ldots,e^{i\varphi_8}\big),
\end{split}
\end{equation}
where $\varphi_i\in[0,2\pi)$ are taken with respect to the flat measure on the eight-torus. 

In this parametrization, the entangling power (\ref{eq:epsilon_definition}) reduces to a relatively short expression
\begin{equation} \label{eq:epsilon_diagonal}
\begin{split}
\epsilon_{1}(D_{\vec{\varphi}})
	= \frac{10}{27}
	- \frac{4}{81}\big(
& c^{14}_{23} + c^{16}_{25} + c^{17}_{35} 
	+ c^{28}_{46} + c^{38}_{47} + c^{58}_{67} \big) \\
	- \frac{1}{81} \big(
& c^{36}_{45} + c^{27}_{45} + c^{27}_{36} 
	+ c^{18}_{45} + c^{18}_{36} + c^{18}_{27} \big),
\end{split}
\end{equation}
where $c^{ij}_{kl}\coloneqq \cos \left(\varphi_i+\varphi_j-\varphi_k-\varphi_l\right)$. Since the average value of the cosine function over the whole period is 0, we can immediately state that the average entangling power of diagonal unitary matrices is equal to the constant term in~the~above expression, $\braket{\epsilon_{1}}_{\mathcal{D}(8)}=\frac{10}{27}\approx 0.37$. Curiously, this number is precisely equal to the entangling power of the Fredkin, Toffoli and the three-qubit controlled sign gates.

Finding the maximum value is a more complex task that requires optimization over $\varphi_i$. To this end, we introduce new variables $\omega_i$, $\delta_j$ and perform the following change of variables:
\begin{equation} \label{eq:diagonal_reparametrization}
\begin{split}
\vec{\varphi}\to \big\{&
	\omega_1,
	\omega_1+\omega_2+\delta_1,
	\omega_3,
	\omega_2+\omega_3,
	-\omega_2-\omega_3+\omega_4+\delta_2,\\
& -\omega_3+\omega_4-\delta_3,
	-\omega_1-\omega_2+\omega_4-\delta_1,
	-\omega_1+\omega_4\big\},
\end{split}
\end{equation}
Note that while this operation reduces the number of parameters from eight to seven, it comes with no loss of generality, as the global phase of unitary gates is irrelevant for physical considerations (including the entangling power). In other words, one of the eight parameters in equation (\ref{eq:diagonal_parametrization}) has been redundant from the start. 

In the new parametrization (\ref{eq:diagonal_reparametrization}), the formula for the entangling power (\ref{eq:epsilon_diagonal}) takes a particularly simple form,
\begin{equation} \label{eq:epsilon_diagonal_deltas}
\begin{split}
\epsilon_{1}(D_{\vec{\delta}})= \frac{1}{81}\bigg[&
	29 - 8\cos\delta_1 - 2\cos\delta_2 - 2\cos\delta_3 
	- 8\cos(\delta_1+\delta_2+\delta_3) \\
& -4\cos(\delta_1+\delta_2)-4\cos(\delta_1+\delta_3)-\cos(\delta_2+\delta_3)\bigg].
\end{split}
\end{equation}
Notably, it depends only on the three parameters $\delta_i$. This should not be surprising: a three-qubit gate can be decomposed into a tensor product of three one-qubit gates if and only if it is possible to assign a definite phase differences between the three subgates. Thus, the entangling power of the gate is proportional to how ``difficult'' it is to assign these three phases, which is  measured by the parameters $\delta_i$.

With just three independent parameters, it is possible to find all the extremal points of the entangling power (\ref{eq:epsilon_diagonal_deltas}), i.e. points in which all of its first derivatives vanish. Since the cube $[0,2\pi]^3\ni(\delta_1,\delta_2,\delta_3)$ is a closed and bounded set, one of the extremal points has to be the global maximum of the function.

Using this method, we find that the maximum of the entangling power over diagonal unitary matrices is equal to $\max_{\mathcal{D}(8)}\epsilon_{1}=16/27\approx 0.59$ and is obtained solely by diagonal unitary matrices of the form
\begin{equation}
\begin{split}
D_{\vec{\omega}}\coloneqq\diag\big[&
	e^{i\omega_1},
	-e^{i(\omega_1+\omega_2)},
	e^{i\omega_3},
	e^{i(\omega_2+\omega_3)},
	e^{-i(\omega_2+\omega_3-\omega_4)},\\
& e^{-i(\omega_3-\omega_4)},
	-e^{-i(\omega_1+\omega_2-\omega_4)},
	e^{-i(\omega_1-\omega_4)}\big],
\end{split}
\end{equation}
for example, the diagonal hermitian matrix
\begin{equation}
\begin{split}
H_{\mathcal{D}(8)}=\diag(1,1,1,-1,1,-1,-1,-1).
\end{split}
\end{equation}

\paragraph{Unitary matrices $\bm{U(8)}$.} In the case of random unitary matrices, using Corollary \ref{th:corollary} we immediately arrive at the mean value $\braket{\epsilon_{1}}_{U(8)}=2/3\approx 0.67$, as well as the maximum $\max_{U(8)}\epsilon_{1}=\tilde{\epsilon}_{1}=8/9\approx 0.89$. An example maximizing unitary matrix is given by the hermitian matrix
\begin{equation} \label{eq:unitary_AME_matrix}
\begin{split}
H_{U(8)}=\frac{1}{\sqrt{2^3}}\begin{bmatrix}
-1 & -1 & -1 & 1 & -1 & 1 & 1 & 1 \\
-1 & -1 & -1 & 1 & 1 & -1 & -1 & -1 \\
-1 & -1 & 1 & -1 & -1 & 1 & -1 & -1 \\
1 & 1 & -1 & 1 & -1 & 1 & -1 & -1 \\
-1 & 1 & -1 & -1 & -1 & -1 & 1 & -1 \\
1 & -1 & 1 & 1 & -1 & -1 & 1 & -1 \\
1 & -1 & -1 & -1 & 1 & 1 & 1 & -1 \\
1 & -1 & -1 & -1 & -1 & -1 & -1 & 1
\end{bmatrix}.
\end{split}
\end{equation}
which arises from direct application of the corollary to the AME(6,2) state provided in eq. (10) in \cite{AME_state}.

To put the things into a wider perspective, let us remind the reader that the amount of entanglement carried by the W state (\ref{eq:W}) is also equal to $8/9$. 
This implies that the action of the maximizing unitary gates
on random separable states produces entangled states with average
$\tau_1$  equal to the entanglement characteristic of the state $\ket{W}$.

\paragraph{Orthogonal matrices $\bm{O(8)}$.} Since the unitary matrix (\ref{eq:unitary_AME_matrix}) maximizing the entangling power over the unitary group $U(8)$ is orthogonal in addition to being unitary, it immediately follows that the maximum entangling power over the orthogonal group is equal to $\max_{O(8)}\epsilon_{1}=\max_{U(8)}\epsilon_{1}=8/9$. As for the mean value, making use of Corollary \ref{th:corollary} once again we find $\braket{\epsilon_{1}}_{O(8)}=208/315\approx 0.66$.

\vspace{5mm}

The results of this section are summarized in Table \ref{tab:gate_properties} and further illustrated in Figure \ref{fig:histogram} showing a probability histogram of the entangling power for ensembles of $8!$ permutation matrices of size eight, the same number of random diagonal unitary matrices, random orthogonal matrices and random unitary matrices distributed according to the Haar measure.

\begin{table}[!tb]
\captionsetup{width=0.95\textwidth}
\caption{Summary of the entangling properties of three-qubit gates, taken from ensembles of: diagonal unitary matrices $\mathcal{D}(8)$, permutation matrices $\mathcal{P}(8)$, random orthogonal matrices generated according to the Haar measure on $O(8)$ and random unitary matrices from $U(8)$.}
\centering
\begin{tabular}{|c|c|c|c|c|} 
 \hline
  & \thead{$\mathcal{D}(8)$} & \thead{$\mathcal{P}(8)$}
  & \thead{$O(8)$} & \thead{$U(8)$} \\
 \hline
 $\min\epsilon_{1}$ & $0$ & $0$ & $0$ & $0$ \\
 $\braket{\epsilon_{1}}$ & $10/27\approx 0.37$
 & $184/315\approx 0.58$ & $208/315\approx 0.66$ & $2/3\approx 0.67$ \\ 
 $\max\epsilon_{1}$ & $16/27\approx 0.59$ 
 & $64/81\approx 0.79$ & $8/9\approx 0.89$ & $8/9\approx 0.89$ \\
 \hline\hline
\end{tabular}
\label{tab:gate_properties}
\end{table}

\begin{figure}[!tb]
	\centering
	\includegraphics[width=1\textwidth]{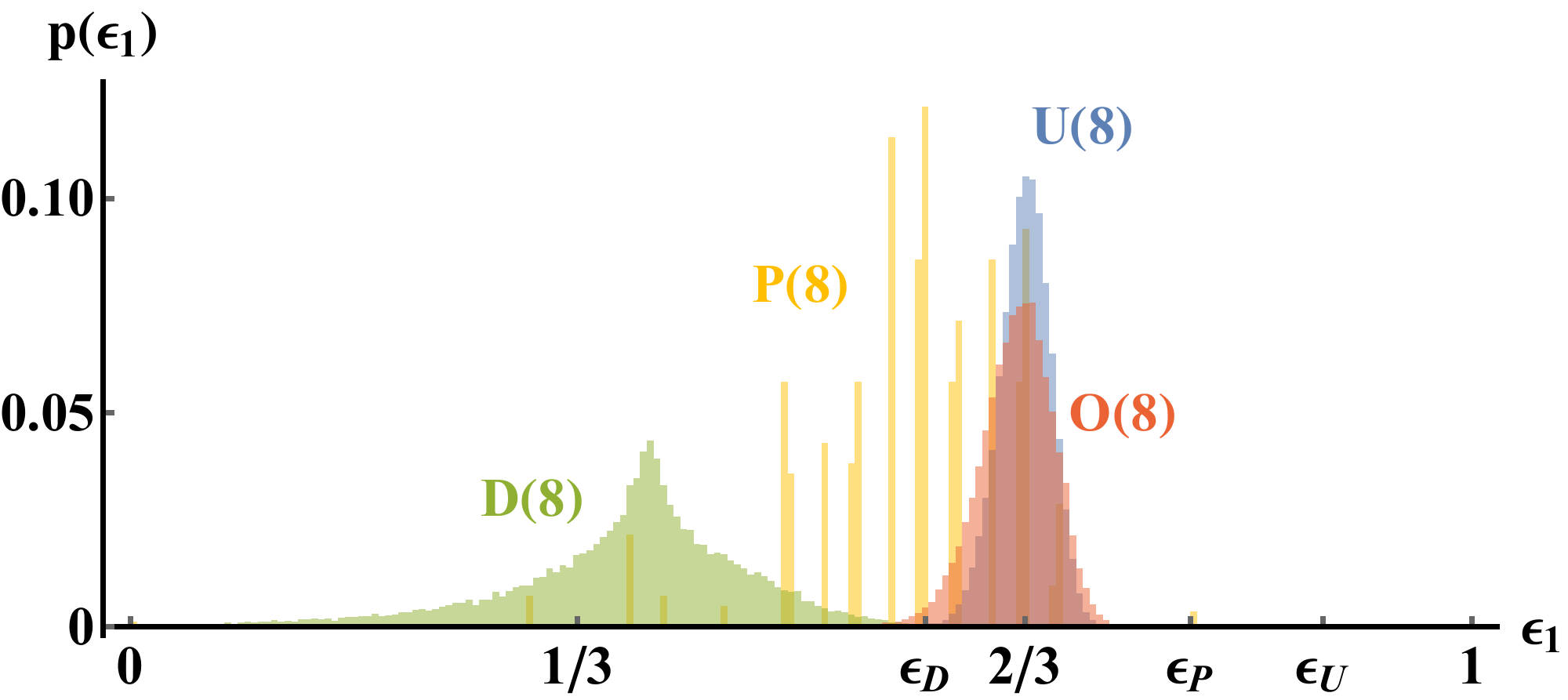}
	\captionsetup{width=1\textwidth}
	\caption{Probability histograms of the entangling power $\epsilon_1$ of all $40\:320$ permutation matrices $P$ of size eight and the same number of diagonal unitary matrices $D$, orthogonal matrices $O$, and generic unitary matrices $U$. In~addition, the maximum values over the respective ensembles have been denoted by $\epsilon_D=16/27$, $\epsilon_P=64/81$, $\epsilon_U=\epsilon_O=8/9$.}
	\label{fig:histogram}
\end{figure}

\section{Concluding remarks}\label{sec:summary}
In this work we have studied the entangling properties of multipartite unitary gates with respect to the chosen measure of entanglement $\tau_1$. We have derived an analytical expression for the entangling power of a tripartite gate as an explicit function of the gate, linking the entangling power of gates acting in tripartite Hilbert space of dimension $d_1d_2d_3$ to the entanglement properties of states in the extended Hilbert space of dimension $(d_1d_2d_3)^2$. Building upon these results, we have computed the mean value of the entangling power of tripartite unitary gates of an arbitrary size and provided an upper bound for the maximum, which we have linked to the AME states in the extended six-party Hilbert space.

These results were then generalized to unitary gates acting on $n$ parties. In~particular, we have found that a gate $U$ acting on $\mathcal{H}_d^{\otimes n}$ which saturates the upper bound for entangling power corresponds to an AME state of $2n$ subsystems with $d$ levels each.

Finally, we have employed our findings to analyze in detail the entangling properties of relevant classes of three-qubit unitaries. We have shown that a generic unitary gate of size $2^3=8$ typically has a slightly larger entangling power than a generic orthogonal gate and much larger entangling power than typical permutation matrices or diagonal unitary matrices.

Based on this work, we propose two main directions for future research. Firstly, due to the existence of two inequivalent classes of maximally entangled three-qubit states, in addition to one-tangle studied in this contribution, there are two more valid measures of three-qubit entanglement. It would be especially interesting to find a formula analogous to ours for so-called \emph{three-tangle}, which measures genuine tripartite entanglement in the state. We stress that due to the way these measures are defined, finding the expression for the entangling power with respect to one of them would immediately yield the expression for the other upon the use of our formula. 

Secondly, the nonlocal properties of a given unitary gate $U$ applied $k$-times, and the influence of the local interlacing dynamics $V_{\textnormal{loc}}$ was recently studied for a bipartite setup \cite{interlacing2,interlacing1,A19}, where the quantity $e_p\big((UV_{\textnormal{loc}})^k\big)$, closely related to the entangling power, was investigated. It would be interesting to extend some of these results also to the multipartite case.

\section*{Acknowledgements}
It is a pleasure to thank Felix Huber, Arul Lakshminarayan and Nikolai Wyderka for fruitful discussions and correspondence. Financial support by Narodowe Centrum Nauki under the grant number DEC-2015/18/A/ST2/00274 is gratefully acknowledged. In addition, Tomasz Linowski would like to acknowledge partial support by the Foundation for Polish Science (IRAP project, ICTQT, contract no. 2018/MAB/5, co-financed by EU within Smart Growth Operational Programme).

\renewcommand\theadalign{bc}
\renewcommand\theadfont{\normalsize}
\renewcommand\theadgape{\Gape[4pt]}
\renewcommand\cellgape{\Gape[4pt]}

\begin{appendices}

\section{Averages over orthogonal group} \label{app:integration}
\setcounter{equation}{0}
\renewcommand{\theequation}{\ref{app:integration}\arabic{equation}}
In order to find second moments of the orthogonal group $O(d)$, it is necessary to find proper orthogonal Weingarten functions, $\langle Wg^O(q),r \rangle$ for two permutations $q$ and $r$.
These were first provided in \cite{orthogonal2006} and further extended for certain setups in \cite{orthogonal_physics,orthogonal_msc}. All the equations in this Appendix are based on these two extensions.

The desired integral can be expanded by these means to the form:
\begin{equation}\label{eq:weingarten_first}
\begin{split}
\int_{O(d)} & dO O^{i_1}_{j_1} O^{i_2}_{j_2} O^{i_3}_{j_3} O^{i_4}_{j_4} = \\
    &=\sum_{q,r \in \{p_1,p_2,p_3\}} \langle Wg^O(q),r \rangle \delta_{i_1}^{i_{q(1)}}\delta_{i_2}^{i_{q(2)}}\delta_{i_3}^{i_{q(3)}}\delta_{i_4}^{i_{q(4)}}\delta_{j_1}^{j_{r(1)}}\delta_{j_2}^{j_{r(2)}}\delta_{j_3}^{j_{r(3)}}\delta_{j_4}^{j_{r(4)}}, 
\end{split}
\end{equation}
where the sum is over all three possible permutations, created by two transpositions: $p_1 = \{(12)(34)\}$, $p_2 = \{(13)(24)\}$ and $p_3 =\{(14)(23)\}$.

The Weingarten functions $\langle Wg^O(q),r \rangle$ are calculated by joining the two permutations $q$ and $r$ into a new permutation $qr$. This new permutation can be uniquely broken into cycles, whose lengths are the only property relevant for our purposes. For each cycle with length $l$ in the structure of $qr$, there is another cycle of length $l$ -- see Lemma 1.16 in \cite{orthogonal_msc}. By taking half of these cycles, the value of $\langle Wg^O(q),r \rangle$ is uniquely determined. 

As an example, we consider $\langle Wg^O(p_1),p_1 \rangle$ with an involution $p_i$ defined above:
    \begin{equation}
        p_1p_1 =  \{(12)(34)\} \{(12)(34)\} =  \{(1)(2)(3)(4)\},
    \end{equation}
so that the joint permutation $p_1p_1$ consists of four cycles of length one. Taking half of these cycles yields two  cycles of length 1, which we write as $[1,1]$.
In \cite{orthogonal_msc} (App. B2) one can find this value:
\begin{equation}
    \langle Wg^O(p_1),p_1 \rangle = Wg^O([1,1],d) = \frac{d+1}{d(d-1)(d+2)}.
\end{equation}
The same result holds for $\langle Wg^O(p_i),p_i \rangle$ with any of the above $p_i$. If the permutations $p_i$ and $p_j$ are different, then the permutation $p_i p_j$ consists of two cycles of length 2 and so
\begin{equation}
    \langle Wg^O(p_i),p_j \rangle = Wg^O([2],d) = \frac{-1}{d(d-1)(d+2)} \text{ for $i\neq j$}.
\end{equation}

Inserting the above formulas into the desired integral (\ref{eq:weingarten_first}), we obtain the required formula
\begin{equation} \label{eq:symbolic_integration_orthog}
\begin{split}
\int_{O(d)} dO O^{i_1}_{j_1} O^{i_2}_{j_2} O^{i_3}_{j_3} O^{i_4}_{j_4}
	= \:&\frac{\left(
	\delta^{i_2}_{i_1}\delta^{i_4}_{i_3}\delta^{j_2}_{j_1}\delta^{j_4}_{j_3}
	+\delta^{i_3}_{i_1}\delta^{i_4}_{i_2}\delta^{j_3}_{j_1}\delta^{j_4}_{j_2}
	+\delta^{i_4}_{i_1}\delta^{i_3}_{i_2}\delta^{j_4}_{j_1}\delta^{j_3}_{j_2}
	\right)}{d(d-1)}\\
& -	\frac{\left(
	\delta^{i_2}_{i_1}\delta^{i_4}_{i_3}\delta^{j_3}_{j_1}\delta^{j_4}_{j_2}
	+\delta^{i_2}_{i_1}\delta^{i_4}_{i_3}\delta^{j_4}_{j_1}\delta^{j_3}_{j_2}
	+\delta^{i_3}_{i_1}\delta^{i_4}_{i_2}\delta^{j_2}_{j_1}\delta^{j_4}_{j_3}
	\right)}{d(d-1)(d+1)}\\
& -	\frac{\left(
	\delta^{i_3}_{i_1}\delta^{i_4}_{i_2}\delta^{j_4}_{j_1}\delta^{j_3}_{j_2}
	+\delta^{i_4}_{i_1}\delta^{i_3}_{i_2}\delta^{j_2}_{j_1}\delta^{j_4}_{j_3}
	+\delta^{i_4}_{i_1}\delta^{i_3}_{i_2}\delta^{j_3}_{j_1}\delta^{j_4}_{j_2}
	\right)}{d(d-1)(d+1)},
\end{split}
\end{equation}
used to derive expression (\ref{eq:epsilon_mean_orthog}).

\section{Three-qubit permutation matrices} \label{app:permutations}
\setcounter{table}{0}
\renewcommand{\thetable}{\ref{app:permutations}\arabic{table}}

There exist $8! = 40320$ permutation matrices of order $2^3 = 8$. For each of them we have found their entangling power $\epsilon_1$ and identified $21$ possible values. Their list and the number of elements in each class are presented in Table \ref{tab:permutations}, analogous to the table presented earlier in \cite{permutations_entangling_power} for the bipartite case of two-qutrit system.

\begin{table}[H] 
\captionsetup{width=0.7\textwidth}
\caption{Classification of~the~entangling power of~three-qubit permutation gates.}
\centering
\begin{tabular}{|c|c||c|c|}
 \hline
 \thead{entangling \\ power $\epsilon_{1}$ \\ times $162$} & \thead{number \\ of elements \\ in class}
 & \thead{entangling \\ power $\epsilon_{1}$ \\ times $162$} & \thead{number \\ of elements \\ in class}
 \\ 
 \hline
 $0$ & 48 & $95$ & 3456\\ 
 $48$ & 288 & $96$ & 4896 \\ 
 $60$ & 864 & $99$ & 2304 \\
 $64$ & 288 & $100$ & 2880 \\
 $72$ & 192 & $104$ & 3456 \\
 $79$ & 2304 & $107$ & 2304 \\
 $80$ & 1440 & $108$ & 3744 \\
 $84$ & 1728 & $111$ & 384 \\
 $87$ & 1536 & $112$ & 1152 \\
 $88$ & 2304 & $128$ & 144 \\ 
 $92$ & 4608 & & \\
 \hline\hline
\end{tabular}
\label{tab:permutations}
\end{table}
\end{appendices}

\newpage
\bibliography{report}{}
\addcontentsline{toc}{chapter}{Bibliography}
\bibliographystyle{custombib}

\end{document}